\documentclass[11pt,letterpaper]{article}

\usepackage{amsmath,amsthm}
\usepackage{xcolor}
\usepackage[colorlinks = true,
            linkcolor = blue,
            urlcolor  = blue,
            citecolor = blue,
            anchorcolor = blue]{hyperref}

\newtheorem{theorem}{Theorem}
\newtheorem{lemma}[theorem]{Lemma}

\newtheorem{proposition}[theorem]{Proposition}

\newtheorem{example}[theorem]{Example}

\usepackage{xspace}
\usepackage{bbm}
\usepackage{bm}
\input{mathlig}
\usepackage{mathtools}
\usepackage{relsize}
\usepackage{mathrsfs}
\usepackage{dsfont}
\DeclarePairedDelimiterX{\inp}[2]{\langle}{\rangle}{#1, #2}
\makeatletter
\newcommand*\bigcdot{\mathpalette\bigcdot@{.5}}
\newcommand*\bigcdot@[2]{\mathbin{\vcenter{\hbox{\scalebox{#2}{$\m@th#1\bullet$}}}}}
\makeatother

\newcommand{\muspace}{\mspace{1mu}}

\DeclareRobustCommand{\scond}{\mathchoice{\muspace\vert\muspace}{\vert}{\vert}{\vert}}
\mathlig{|}{\scond}

\DeclareRobustCommand{\discint}{\mathchoice{\mspace{-1.5mu}:\mspace{-1.5mu}}{\mspace{-1.5mu}:\mspace{-1.5mu}}{:}{:}}
\mathlig{::}{\discint}
\newcommand{\suchthat}{\mathchoice{\colon}{\colon}{:\mspace{1mu}}{:}}

\newcommand{\Ac}{\mathcal{A}}
\newcommand{\Bc}{\mathcal{B}}

\newcommand{\Gc}{\mathcal{G}}
\newcommand{\Hc}{\mathcal{H}}

\newcommand{\Pc}{\mathcal{P}}

\newcommand{\Zc}{\mathcal{Z}}

\newcommand{\Xv}{{\bf X}}

\newcommand{\xv}{{\bf x}}

\newcommand{\Pb}{{\mathbf P}}

\newcommand{\xb}{{\mathbf x}}

\newcommand{\Zt}{{\tilde{Z}}}

\newcommand{\kt}{{\tilde{k}}}

\newcommand{\zt}{{\tilde{z}}}
\newcommand{\gt}{{\tilde{g}}}

\def\b{\beta}

\def\e{\epsilon}
\def\eps{\epsilon}
\def\th{\theta}

\DeclareMathOperator\E{\mathsf{E}}
\let\P\relax
\DeclareMathOperator\P{\mathsf{P}}

\newcommand\eg{e.g.,\xspace}
\newcommand\ie{i.e.,\xspace}
\def\textiid{i.i.d.\@\xspace}
\newcommand\iid{\ifmmode\text{ i.i.d. } \else \textiid \fi}

\newcommand{\Real}{\mathbb{R}}

\newcommand{\ones}{\mathds{1}}

\def\mathllap{\mathpalette\mathllapinternal}
\def\mathllapinternal#1#2{%
  \llap{$\mathsurround=0pt#1{#2}$}}

\def\clap#1{\hbox to 0pt{\hss#1\hss}}
\def\mathclap{\mathpalette\mathclapinternal}
\def\mathclapinternal#1#2{%
  \clap{$\mathsurround=0pt#1{#2}$}}

\let\oldstackrel\stackrel
\renewcommand{\stackrel}[2]{\oldstackrel{\mathclap{#1}}{#2}}

\DeclarePairedDelimiterX{\infdivx}[2]{(}{)}{%
  #1\;\delimsize\|\;#2%
}

\def\dtv{\delta}

\renewcommand{\hbar}{h\mathllap{\overline{\vphantom{h}\hphantom{\rule{4.6pt}{0pt}}}\mspace{0.77mu}}}

\catcode`~=11 
\newcommand{\urltilde}{\kern -.06em\lower -.06em\hbox{~}\kern .02em}
\catcode`~=13 

\newcommand{\indi}[1]{\ones_{\{#1\}}}

\DeclarePairedDelimiterX{\norm}[1]{\lVert}{\rVert}{#1}
\DeclarePairedDelimiterX{\abs}[1]{\lvert}{\rvert}{#1}

\usepackage{xparse}

\newcommand*\diff{\mathop{}\!\mathrm{d}}

\let\oldpartial\partial
\renewcommand*{\partial}{\mathop{}\!\oldpartial}
\newcommand{\defeq}{\mathrel{\mathop{:}}=}

\usepackage{mathrsfs}
\allowdisplaybreaks

\usepackage{thm-restate}

\usepackage{amsfonts}
\usepackage{dsfont}
\newcommand\numberthis{\addtocounter{equation}{1}\tag{\theequation}}

\renewcommand{\b}{\beta}

\newcommand{\Lap}{\mathsf{L}}

\newcommand{\RegPor}{\mathsf{Reg}^{\mathsf{port}}}
\newcommand{\ERegPor}{\overline{\mathsf{Reg}}^{\mathsf{port}}}
\newcommand{\RegProb}{\mathsf{Reg}^{\mathsf{prob}}}
\newcommand{\qkt}{q_{\mathsf{KT}}}

\usepackage[utf8]{inputenc} 
\usepackage[T1]{fontenc}
\usepackage{graphicx,colordvi,psfrag}
\usepackage{amsmath,amssymb}
\usepackage{epstopdf}
\usepackage{calc,pstricks, pgf, xcolor}
\usepackage{nicefrac}
\usepackage{enumerate}
\usepackage{dsfont}
\usepackage{bm}
\usepackage{color}
\usepackage{bbold}
\usepackage{setspace}
\usepackage{amsmath}

\newcommand{\nt}{\Tilde{n}}

\newcommand{\thb}{\boldsymbol{\theta}}
\renewcommand{\th}{\theta}
\renewcommand{\dtv}{d_{\mathsf{TV}}}
\DeclareMathOperator{\dham}{d_H}
\usepackage{cuted}
\usepackage[margin=1in]{geometry} 
\usepackage[round]{natbib}

\title{On Universal Portfolios with Continuous Side Information}

\author{
Alankrita Bhatt$^*$\\
UC San Diego\\
\href{mailto:a2bhatt@ucsd.edu}{a2bhatt@ucsd.edu}
\and J. Jon Ryu\thanks{Both authors contributed equally to this work.}\\
UC San Diego\\
\href{mailto:jongharyu@ucsd.edu}{jongharyu@ucsd.edu}
\and Young-Han Kim\\
UC San Diego/Gauss Labs Inc.\\
\href{mailto:yhk@ucsd.edu}{yhk@ucsd.edu}
}

\date{\vspace{-1em}}

\begin{document}

\maketitle

\begin{abstract}
A new portfolio selection strategy that adapts to a continuous side-information sequence is presented, with a universal wealth guarantee against a class of state-constant rebalanced portfolios with respect to a state function that maps each side-information symbol to a finite set of states.
In particular, given that a state function belongs to a collection of functions of finite Natarajan dimension, the proposed strategy is shown to achieve, asymptotically to first order in the exponent, the same wealth as the best state-constant rebalanced portfolio with respect to the best state function, chosen in hindsight from observed market.
This result can be viewed as an extension of the seminal work of Cover and Ordentlich (1996) that assumes a single state function.
\end{abstract}



\newcommand{\revision}[1]{\textcolor{blue}{#1}}

\section{Introduction}
We study the classical problem of portfolio selection, formally defined as follows.
Suppose that there exist $m\ge 2$ stocks in a stock market and let $\xb_t=(x_{t1},\ldots,x_{tm})\in \Real_{\ge 0}$ denote a market vector at time $t$, which encodes the \emph{price relatives} of stocks on that day. 
That is, for each stock $i\in[m] \defeq \{1,\ldots,m\}$, $x_{ti}\ge 0$ is the ratio of the end price to the start price on day $t$.
Concretely, an investment strategy $a$, at each day $t$, outputs a nonnegative weight vector $a(\cdot|\xv^{t-1})\in\Delta^{m-1}$ over the stocks $[m]$, upon which the investor distributes her wealth accordingly; hereafter, we use $\Bc\defeq\Delta^{m-1}\defeq\{(\th_1,\ldots,\th_m)\in\Real_{\ge 0}^m\suchthat \sum_{i=1}^m \th_i=1\}$ to denote the standard $m$-simplex.
That is, the multiplicative wealth gain on day $t$ (\ie the ratio of wealth on day $t$ to the wealth on day $t-1$) is $\sum_{j\in[m]} a(j|\xv^{t-1})x_{tj}$. Thus, her cumulative wealth gain after $n$ days becomes
\begin{align*}
S_n(a,\xb^n)
&\defeq \prod_{t=1}^n \sum_{j\in[m]} a(j|\xb^{t-1}) x_{tj}
\numberthis
\label{eq:portfolio_wealth_product}
=\sum_{y^n\in[m]^n} \Bigl(
\prod_{t=1}^n a(y_t|\xb^{t-1})\Bigr)\xb(y^n),
\end{align*}
where $\xb(y^n)\defeq x_{1y_1}\cdots x_{ny_n}$ denotes the wealth gain of an extreme investment strategy that puts all money to the stock $y_t$ on day $t$, and the second equality follows from the distributive law.

An investor's goal is to design an investment strategy that maximizes her cumulative wealth $S_n(a,\xv^n)$. 
For a stock market where $\xv^n$ are \iid, it is known that the log-optimal portfolio $\thb^\star$ that maximizes $\E[\log \thb^T \Xv]$ is asymptotically and competitively optimal. A similar result is well-established for stationary ergodic markets, see, \eg \cite[Chapter~16]{Cover--Thomas2006}.
The log-optimal portfolio theory with stochastic market assumptions, however, is unrealistic, as modeling a stock market could be harder than predicting the market. 

As a more realistic alternative, \citet{Cover1991} presented \emph{universal portfolios} that asymptotically achieve the best wealth, to first order in the exponent, attained by a certain class of reference portfolios, with \emph{no statistical assumptions} on the stock market.
For the reference class, Cover considered a class of constant rebalanced portfolios (CRPs), where a CRP parameterized by a weight vector $\thb\in\Bc$ is defined to redistribute its wealth according to $\thb$ on every day. Note that CRPs are optimal in an \iid stock market when the distribution is known.

Later, \citet{Cover--Ordentlich1996} extended the theory to a setup where a discrete side information sequence is causally available to an investor; in practice, the side information sequence can be thought to encode an external information 
that may help predict the stock market. 
They proposed a variation of \citeauthor{Cover1991}'s universal portfolios that asymptotically achieves the best wealth attained by a class of \emph{state-wise} CRPs that may play different weight vectors according to the side information.

Taking one step further, in this paper, we consider a more challenging scenario in which a side information sequence $z^n\in\Zc^n$ is continuous-valued, which could even be the (truncated) market history itself.
A reference portfolio we aim to compete with is parameterized by a state-wise CRP and a \emph{state function} $g\suchthat\Zc\to[S]$ for some $S\ge 2$ and  plays the state-wise CRP according to the state sequence $g(z^n)\defeq g(z_1)\ldots g(z_n)$, where we assume a class of state functions $\Gc$ from which $g$ is drawn; note that larger the $\Gc$, the richer the reference class.
This flexibility in the class $\Gc$ and the choice of continuous side information sequence may hugely enlarge the capacity of the competitor class since it can capture a variety of investment strategies. As a simple example, consider a portfolio strategy that selects the state based on whether the price relative of the first stock yesterday $\xv_{t,1} \ge \b$ or not for a (variable) threshold $\b$. This falls into this enlarged class with $z_t = \xv_{t-1}$.

As the main result, we propose a new investment strategy that asymptotically achieves the same wealth attained by the best state-constant rebalanced portfolios with a state function drawn from a class of functions of finite Natarajan dimension, under a mild regularity condition on the stochasticity of the side information sequence $Z^n$.
The proposed strategy is based on a generalization of a universal probability assignment scheme recently proposed by \citet{Bhatt--Kim2021}. Note that we assume no transaction costs and that the investor's actions do not affect the market.

The rest of the paper is organized as follows.
In Section~\ref{sec:review}, we review universal portfolios without and with discrete side information, highlighting the connection between universal compression (or probability assignment) and universal portfolios. Section~\ref{sec:main} described the proposed algorithm and a crude approximation algorithm for its simulation, together with some concrete examples of side information sequence. 
We present the proof of the main theorem in Section~\ref{sec:proofs}.
We conclude with discussing related work in Section~\ref{sec:related}. All deferred proofs can be found in the Appendix.

\section{A Review of Universal Portfolio Theory}
\label{sec:review}
\subsection{Universal Portfolios}
In his seminal work, \citet{Cover1991} set an ambitious goal that aims to design an investment strategy $b$ to compete with the best strategy in a class $\Ac$ of investment strategies for any stock market $\xv^n$, in the sense that it minimizes the worst-case regret
\[
\RegPor_n(b,\Ac)\defeq\sup_{\xv^n}\sup_{a\in\Ac}\log\frac{S_n(a,\xv^n)}{S_n(b,\xv^n)}.
\]
We call a portfolio $b$ \emph{universal with respect to $\Ac$} if $\RegPor_n(b,\Ac)=o(n)$, \ie in words, $b$ achieves the same exponential wealth growth rate attained by the best strategy in $\Ac$ chosen in hindsight with observed market.


Remarkably, Cover constructed a universal portfolio with respect to the class of CRPs and established its universality. 
Cover's theory is based on the key observation that competing against CRPs in portfolio optimization is equivalent to competing against \iid Bernoulli models in log-loss prediction problem.
In what follows, we describe this relationship in a general form beyond between \iid probabilities and CRPs.

For any sequential probability assignment scheme $q(\cdot|y^{t-1}) \in \Bc$ (where $y_i \in [m]$) the \emph{probability induced portfolio} $a=\phi(p)$ is defined as
\[
a(j|\xv^{t-1}) \defeq \frac{\sum_{y^{t-1} \in [m]^{t-1}} p(y^{t-1}j)\xv(y^{t-1})}{\sum_{y^{t-1} \in [m]^{t-1} } p(y^{t-1})\xv(y^{t-1})}.
\numberthis
\label{eq:def_prob_induced_portfolio}
\]
Note that if $p$ is an \iid probability, \ie $p(\cdot|y^{t-1}) = \thb \in \Bc$, it is easy to check from the expression~\eqref{eq:def_prob_induced_portfolio} that the corresponding portfolio $\phi(p)$ is the CRP parameterized by $\thb$; thus the class of CRPs $\Ac^{\mathsf{CRP}}$ is $\phi(\Pc^{\otimes})$, where we use $\Pc^{\otimes}$ to denote the class of \iid probabilities.

A peculiar property of a probability induced portfolio $a=\phi(p)$ is that the daily gain can be written as
\[
\sum_{y_t\in[m]} a(y_t|\xv^{t-1}) \xv_t(y_t) = \frac{\sum_{y^t} p(y^t)\xv(y^t)}{\sum_{y^{t-1}} p(y^{t-1})\xv(y^{t-1})},
\]
and thus by telescoping, the cumulative wealth gain~\eqref{eq:portfolio_wealth_product} becomes
\begin{align}
S_n(\phi(p),\xv^n) = \sum_{y^n \in [m]^n} p(y^n)\xv(y^n).
\label{eq:wealth_achieved_by_prob_induced_portfolio}
\end{align}
In view of this expression, a probability induced portfolio can be interpreted as a \emph{fund-of-funds}, \ie a mixture of the extremal portfolios with weights $p(y^n)$. 

As alluded to earlier, there is an intimate connection between the portfolio optimization with respect to a class of probability induced portfolios and the corresponding log-loss prediction problem.
In the log-loss prediction problem, given a class of probabilities $\Pc$, we define the worst-case regret of a probability $q$ with respect to $\Pc$ as
\begin{align} \label{eq:RegProbDefn}
\RegProb_n(q,\Pc)=\sup_{y^n} \sup_{p\in\Pc} \log \frac{p(y^n)}{q(y^n)}
\end{align} 
and call a probability $q$ \emph{universal} with respect to $\Pc$ if $\RegProb_n(q,\Pc)=o(n)$.
The following proposition shows that the portfolio optimization with respect to $\phi(\Pc)$ is no more difficult than the corresponding log-loss prediction problem with respect to $\Pc$.
\begin{proposition}
\label{prop:portfolio_is_not_harder}
For any probability $q$ and any class of probability assignments $\Pc$, we have \[\RegPor_n(\phi(q),\phi(\Pc)) \le \RegProb_n(q,\Pc).\]
\end{proposition}
\begin{proof}
We first recall \eqref{eq:wealth_achieved_by_prob_induced_portfolio} that
the cumulative wealth of the probability induced portfolio $\phi(p)$ is written as $S_n(\phi(p),\xb^n)=\sum_{y^n} p(y^n)\xb(y^n)$.
Hence, for any probability $q$, we can write
\begin{align*}
\RegPor_n(\phi(q), \phi(\Pc))
&=\sup_{\xv^n}\sup_{p\in\Pc} \frac{S_n(\phi(p),\xb^n)}{S_n(\phi(q),\xb^n)}\\
&=\sup_{\xv^n}\sup_{p\in\Pc}\frac{\sum_{y^n} p(y^n)\xb(y^n)}{\sum_{y^n} q(y^n)\xb(y^n)}\\
&\stackrel{(a)}{\le}
\sup_{p\in\Pc}\max_{y^n} \frac{p(y^n)}{q(y^n)}
= \RegProb_n(q,\Pc),
\end{align*}
where $(a)$ follows by Lemma~\ref{lem:ratio} below.
\end{proof}
\begin{lemma}[\cite{Cover--Thomas2006}, Lemma 16.7.1]
\label{lem:ratio}
Let $a_1,\ldots,a_n,b_1,\ldots,b_n$ be nonnegative real numbers. Then, defining $0/0=0$, we have
\begin{align*}
\frac{\sum_{i=1}^n a_i}{\sum_{i=1}^n b_i} \le \max_{j=1,\ldots,n}\frac{a_j}{b_j}.
\end{align*}
\end{lemma}

A direct implication of this statement is that if a probability assignment $q$ is universal with respect to $\Pc$ for the log-loss prediction problem, then the induced portfolio $\phi(q)$ is universal with respect to $\phi(\Pc)$.
If we consider the class of all \iid probabilities $\Pc^{\otimes}$, it is well known that the Laplace probability assignment $q_{\Lap}(y^n)\defeq \int_{\Bc} \mu(\thb)p_{\thb}(y^n)\diff \thb$ is universal for $\Pc^{\otimes}$, where $\mu(\thb)$ is the uniform density over $\Bc$ and $p_{\thb}(y^n)$ is the \iid probability with parameter $\thb=(\th_1,\ldots,\th_m)\in\Bc$, \ie $p_{\thb}(y^n)\defeq \prod_{i=1}^n \th_{y_n}=\prod_{j=1}^m \th_j^{k_j}$ with $k_i = |\{t\suchthat y_t = i\}|$.\footnote{We remark that while the Krichevsky--Trofimov (KT) probability assignment $\qkt$ is universal with an optimal constant in the regret, we consider $q_{\Lap}$ for simplicity throughout this paper.
}
Indeed, we have:
\begin{lemma}[\cite{Cesa-Bianchi--Lugosi2006}, Chapter 9] \label{lem:LaplaceProbAssgnGuarantee}
\[
\sup_{\thb \in \Bc} \sup_{y^n \in [m]^n} \log \frac{p_{\thb}(y^n)}{q_{\Lap}(y^n)} \le m \log n.
\]
\end{lemma}
Hence, $\phi(q_{\Lap})$ is a universal portfolio for $\Ac^{\mathsf{CRP}}=\phi(\Pc^{\otimes})$---this is \citeauthor{Cover1991}'s universal portfolio.
We remark that the universal portfolio $\phi(q_{\Lap})$ can be expressed as
\[
\phi(q_{\Lap})(\cdot|\xv^{t-1})=\frac{\int_{\Bc}\thb S_{t-1}(\thb,\xv^{t-1})\mu(\thb)\diff\thb}{\int_{\Bc} S_{t-1}(\thb,\xv^{t-1})\mu(\thb)\diff\thb},
\]
and is thus also known as the $\mu$-weighted portfolio.

\subsection{Universal Portfolios with Discrete Side Information}
Let us now consider a scenario at each time $t$, the investor is additionally given a discrete side information $w_t\in [S]$ for some $S\ge 1$ and chooses a portfolio $a(\cdot|\xv^{t-1};w^t)\in\Bc$, as considered by~\citet{Cover--Ordentlich1996}. 
Since the investor's multiplicative wealth gain is $\sum_{y\in[m]} a(y|\xv^{t-1};w^t) \xv_t(y)$, 
similar to the no-side-information setting, the cumulative wealth factor is
\begin{align*}
S_n(a,\xb^n;w^n)
&\defeq \prod_{t=1}^n \sum_{j\in[m]} a(j|\xb^{t-1};w^t) x_{tj}
\numberthis
\end{align*}
and we define the worst-case regret as
\[
\RegPor_n(b,\Ac;w^n)\defeq\sup_{a\in\Ac}\sup_{\xv^n}\log\frac{S_n(a,\xv^n;w^n)}{S_n(b,\xv^n;w^n)}
\]
for a class $\Ac$ of portfolios  that also adapt to $w^n$.
Concretely, as a natural extension of CRPs, we consider a class of state-constant rebalanced portfolios (state-CRPs), denoted as $\Ac_S^{\mathsf{CRP}}$, where a state-CRP parameterized by a $S$-tuple $(\thb_1,\ldots,\thb_S)\in\Bc^S$ plays a portfolio $\thb_{w_t}$ at each time $t$.

Paralleling the connection between probability and portfolio in the no-side-information case, we can also define a probability induced portfolio in this setting. In the log-loss prediction with a causal side information sequence, a learner is asked to assign a probability $p(\cdot|y^{t-1};w^t)$ over $[m]$ based on the causal information, \ie past sequence $y^{t-1}$ and the side information sequence $w^t$. Here, we use $p(y^n\|w^n)\defeq \prod_{t=1}^n p(y_t|y^{t-1};w^t)$
to denote the joint probability over $y^n$ given $w^n$.
The probability induced portfolio $a=\phi(p)$ is then defined as
\[
a(j|\xv^{t-1};w^t) \defeq \frac{\sum_{y^{t-1}} p(y^{t-1}j\|w^{t})\xv(y^{t-1})}{\sum_{y^{t-1}} p(y^{t-1}\|w^{t-1})\xv(y^{t-1})},
\numberthis
\label{eq:def_prob_induced_portfolio_with_side_info}
\]
and as in the no-side information setting, we can write
\[
S_n(\phi(p),\xv^n;w^n) = \sum_{y^n} p(y^n\|w^n)\xv(y^n).
\label{eq:wealth_achieved_by_prob_induced_portfolio_side_info}
\]
For example, the class of $S$-state-CRPs $\Bc_S^{\mathsf{CRP}}$ is induced by the class of all $S$-state \iid probabilities  $\Pc_S^{\otimes}$, \ie $\Bc_S^{\mathsf{CRP}}=\phi(\Pc_S^{\otimes})$. To see this, note that every $S$-state-CRP parameterized by $\thb_{1:S}=(\thb_1,\ldots,\thb_S)$ is the portfolio induced by the state-wise \iid probability assignment
$p_{\thb_{1:S}}(y^n\|w^n) 
\defeq \prod_{t=1}^n p_{\thb_{w_t}}(y_t)$.

Moreover, as stated in Proposition~\ref{prop:portfolio_is_not_harder}, solving the log-loss prediction problem suffices for the probability optimization with side information with respect to a class of probability induced portfolios.
\begin{proposition}
\label{prop:portfolio_probability_side_information}
For any probability assignment $q$ and any class of probability assignment schemes $\Pc$ with side information sequence $w^n$, we have $\RegPor_n(\phi(q),\phi(\Pc);w^n)\le \RegProb_n(q,\Pc;w^n)$, where we define
\[
\RegProb_n(q,\Pc;w^n)\defeq
\sup_{p\in\Pc} \max_{y^n} \log \frac{p(y^n\|w^n)}{q(y^n\|w^n)}.
\]
\end{proposition}
Note that for the class of $S$-state-wise \iid distributions $\Pc_S^{\otimes}$, the state-wise extension of the Laplace probability assignment $q_{\Lap;S}$ that assigns 
\begin{align}\label{eq:LaplaceStatewiseDefn}
q_{\Lap;S}(y^n\|w^n) \defeq \prod_{s=1}^S q_{\Lap}(y^n(s;w^n)),
\end{align}
where $y^n(s;w^n)=(y_i\suchthat w_i=s,i\in[n])$, is universal, and so $\phi(q_{\Lap;S})$ is universal for $\Ac_S^{\mathsf{CRP}}=\phi(\Pc_S^{\otimes})$---this is \citeauthor{Cover--Ordentlich1996}'s universal portfolio.

\section{Main Results}
\label{sec:main}

\subsection{Universal Portfolios with Continuous Side Information}
We now consider our main setting where a side information sequence $z^n\in\Zc^n$ is continuous-valued.
For example, in this setup, one may take $z_t$ as a suffix of the market history $\xv_{t-k}^{t-1}$ for some $k\ge 1$. 
As described earlier in the introduction, we aim to design a universal portfolio that competes against a class of state-CRPs that adapts to the sequence $g(w^n)$, where $g$ is a state function $g\suchthat \Zc\to [S]$ assumed to belong to a class of functions $\Gc$. Note that a singleton $\Gc=\{g\}$ recovers the setting of \citet{Cover--Ordentlich1996}.
Our goal is to design a portfolio that is universal for a largest possible $\Gc$ with a minimal assumption on the side information sequence.
In this paper, we will assume that the \emph{Natarajan dimension}~\citep{Shalev-Shwartz--Ben-David2014} of $\Gc$, denoted as $\mathrm{Ndim}(\Gc)$, is finite. The Natarajan dimension can be seen as a generalization of the classic VC dimension, when the function class under consideration is not binary. 

Leveraging the established connection between probability and portfolio, we continue to view the class of state-wise CRPs $\Bc_S^{\mathsf{CRP}}=\phi(\Pc_S^{\otimes})$ as the class of portfolios induced by $\Pc_S^{\otimes}$ and describe the problem in an abstract setting.
For a class of probability induced portfolios with (discrete) side information $\Ac=\phi(\Pc)$ and a class of state functions $\Gc$, our goal is to design a strategy $b$ that achieves a sublinear worst-case regret
\[
\RegPor_n(b;\Ac,\Gc;\xv^n,z^n) \defeq \sup_{g\in\Gc}\sup_{a\in\Ac} 
 \log\frac{S(a,\xv^n;g(z^n))}{S(b,\xv^n;z^n)}.
\]
Similar to the universal portfolios with discrete side information, a universal portfolio can be readily induced by a universal probability with respect to a continuous side information sequence with an unknown state function, based on the following statement.
\begin{proposition} \label{thm:PortfolioIsEasierSideInfo}
For any $\xv^n$ and $z^n$, we have
\[
\RegPor_n(\phi(q);\phi(\Pc),\Gc;\xv^n,z^n) \le \RegProb_n(q;\Pc,\Gc;z^n),
\]
where
\[
\RegProb_n(q;\Pc,\Gc;z^n)\defeq \sup_{g\in\Gc} \sup_{p\in\Pc} \max_{y^n} \log \frac{p(y^n\|g(z^n))}{q(y^n\|z^n)}.
\]
\end{proposition}
In this work, we specifically plug-in an extended version of the universal probability assignment $q_{\Gc}^*$ proposed by~\cite{Bhatt--Kim2021}, which was designed for $m=2,S=2$ with regret guarantee established when $y^n$ is random and the side information sequence $Z^n$ is \iid. We will extend their scheme for arbitrary $m$ and $S$ with a guarantee for adversarial $y^n$ and non-\iid $Z^n$.


Below, we further assume that a side information sequence $Z^n$ is stochastic with distribution $P_{Z^n}$ which may be arbitrarily correlated with the stock market $\Xv^n$; the universality is established with respect to the expected worst-case regret
\[
\ERegPor_n(b;\Ac,\Gc)\defeq \E\bigl[\RegPor_n(b;\Ac,\Gc;\Xv^n,Z^n)\bigr],
\]
where the expectation is over a joint distribution $\Pb_{\Xv^n,Z^n}$. 
We remark that it is unclear whether the required stochastic assumptions on $Z^n$ in Theorem~\ref{thm:AsymptoticConsistency} are an artifact of our analysis or whether they can be completely removed and universality can be established for individual sequences $z^n$. 
We leave this question for future work; see also Section~\ref{sec:related}. 

\paragraph{Proposed Strategy}

Firstly, for any $\nt \in \mathbb{N}$ and any $\zt^{\nt} \in \Zc^{\nt}$,  let $\{\gt_1,\dotsc,\gt_{\ell}\} \subset \Gc$ be a \emph{minimal empirical covering} of $\Gc$ with respect to $\zt^{\nt}$, \ie a set of functions such that 
$\{\gt_i(\zt^{\nt})\suchthat i \in [\ell]\} = \{g(\zt^{\nt})\suchthat g \in \Gc\}$ with the minimum possible size $\ell=\ell(\zt^{\nt})$. 
Then, we define a mixture probability assignment
\[
q_{\Gc;\zt^{\nt}}(y^i\|z^i) \defeq \frac{1}{\ell}\sum_{j=1}^{\ell} q_{\Lap;S}(y^i\|\gt_j(z^i))
\numberthis\label{eq:empirical_covering_mixture}
\]
with respect to the empirical covering,
and define the induced sequential probability assignment 
\[
q_{\Gc;\zt^{\nt}}(y_i|y^{i-1};z^i) \defeq \frac{q_{\Gc;\zt^{\nt}}(y^i\|z^i)}{q_{\Gc;\zt^{\nt}}(y^{i-1}\|z^{i-1})}.
\]
The proposed probability assignment $q_{\Gc}^*$ is then defined as follows. First, we split the $n$ time steps into $\lceil \log_2 n \rceil$ epochs: starting from $j = 1$, define the $j$-the epoch to consist of the time steps $2^{j-1} + 1 \le i \le 2^j$. So, the first epoch consists of $z_2$, the second epoch consists of $z_3^4$, the third epoch consists of $z_5^8$ and so on. Then,
\begin{itemize} \itemsep 0.2em
\item For $i=1$, $q_{\Gc}^*(\cdot|z_1) \defeq 1/m$;
\item For $i \ge 2$, if $2^{j-1} + 1 \le i \le 2^j$, \ie if the time step $i$ falls within the $j$-th epoch, then
\begin{align*} 
     q_\Gc^*(y_i|y^{i-1};z^i) \defeq \frac{q_{\Gc;z^{2^{j-1}}}(y_{2^{j-1}+1}^{i}\|z_{2^{j-1}+1}^i)}{q_{\Gc;z^{2^{j-1}}}(y_{2^{j-1}+1}^{i-1}\|z_{2^{j-1}+1}^{i-1})},
\end{align*}
where we define $q_{\Gc;z^{2^{j-1}}}(\emptyset\|\emptyset)=1$ by convention.
\end{itemize}
Concretely, the probability assigned over $y^n$ given $z^n$ for some $n\in (2^{J-1}, 2^J]$ is
\begin{align*}
q_{\Gc}^*(y^n\|z^n) 
&= \prod_{i=1}^n q_{\Gc}^*(y_i|y^{i-1};z^i)\\
&= q_{\Gc;\emptyset}(y_1\|z_1) q_{\Gc;z_1}(y_2\|z_2) q_{\Gc;z^2}(y_3^4\|z_3^4) 
\cdots 
q_{\Gc;z^{2^{J-1}}}(y_{2^{J-1}+1}^{n}\|z_{2^{J-1}+1}^{n}).
\numberthis\label{eq:assigned_probability}
\end{align*}
Finally, we obtain a sequential portfolio $a = \phi(q_\Gc^*)$ via the expression~\eqref{eq:def_prob_induced_portfolio_with_side_info}.

While the main focus of this paper is to construct a provably universal portfolio with continuous side information, we also include a discussion on its simulation in Appendix~\ref{sec:simulation}.

\subsection{Performance Guarantee and Examples}
Given a class of $S$-state functions $\Gc$, we need to impose a structural condition on the side information sequence $Z^n\sim P_{Z^n}$ as a stochastic process.
For any binary function class $\Hc \subset \{\Zc \to \{0,1\}\}$, we define 
\begin{align}\label{eq:rhodefn}
    \rho_{\Hc}(Z^n) = \sup_{h \in \Hc}\Bigl| \sum_{i=1}^n \bigl(h(Z_i) - \E[h(Z_i)]\bigr)\Bigr|,
\end{align}
which is a well-studied quantity in the empirical process theory.
Specifically, we are interested in the binary function class $\indi{\Gc\times\Gc}\defeq \{h\suchthat\Zc\to\{0,1\}\suchthat h(z)=\indi{(g(z)\neq g'(z))} \text{ for $g,g'\in\Gc$}\}$.
With a slight abuse of notation, we use $\rho_{\Gc\times\Gc}(Z^n)$ to denote $\rho_{\indi{\Gc\times\Gc}}(Z^n)$.

Now we can state our main result.



\begin{theorem}[Asymptotic universality]
\label{thm:AsymptoticConsistency}
For any collection of functions $\Gc$ of finite Natarajan dimension and any stationary stochastic process $Z^n$ such that \begin{align}\label{eq:ConsistencyCondition}
\E[\rho_{\Gc \times \Gc}(Z^n)]=o\Bigl(\frac{n}{\log ^2 n}\Bigr),
\end{align}
the induced portfolio $\phi(q_{\Gc}^*)$ satisfies
\[
\lim_{n\to\infty}\frac{1}{n} \ERegPor(\phi(q_{\Gc}^*), \Ac_S^{\mathsf{CRP}}, \Gc) = 0.
\]
\end{theorem}

In Theorem~\ref{thm:AsymptoticConsistency}, the condition $\E[\rho_{\Gc \times \Gc}(Z^n)] \ll \frac{n}{\log^2 n}$ on the marginal distribution $P_{Z^n}$ is crucial in ensuring consistency of the portfolio $\phi(q_{\Gc}^*)$.
We now provide a few example cases of side information sequences $Z^n$ where this requirement is satisfied.
\begin{example}[\iid processes]\label{example:iidprocesses}
When the joint distribution $P_{\Xv^n,Z^n}$ is such that $Z^n$ is \iid, it is well known that $\E[\rho_{\Hc}(Z^n)] \le C\sqrt{\text{VCdim}(\Hc)n}$ (where $C$ is an absolute constant) for any binary class $\Hc$ and distribution $P_{Z^n}$; see~\citet[Theorem 8.3.23]{Vershynin2018}. 
Following the same logic\footnote{The only change to be made in the proof is in the growth function---rather than $\left(\frac{en}{d}\right)^d$, the growth function in this case is $\le (S^2n)^{2d}$ by Natarajan's Lemma; see Section~\ref{subsec:NonCausalSideInfo}.}, it can be shown that  $\E[\rho_{\Gc \times \Gc}(Z^n)] \le C\sqrt{(d\log S) n}$ and consequently $\ERegPor$ is sublinear; in fact $\ERegPor = \widetilde{O}(\sqrt{n})$.
\end{example}

\begin{example}[$\b$-mixing processes]
\label{ex:beta_mixing}
The quantity $\E[\rho_{\Hc}(Z^n)]$ has also been studied for classes beyond \iid sequences---in particular, \citet{Yu1994} studied the case when $Z^n$ is \emph{$\b$-mixing}, which we now define. For the sigma-fields $\sigma_l \defeq \sigma(Z_1,\dotsc,Z_{\ell})$ and $\sigma_{l+k}' \defeq \sigma(Z_{\ell+k},Z_{\ell+k+1},\dotsc,)$, we define 
\begin{align*}
    \beta_{k} \defeq \frac{1}{2}\sup\{\E|P(B|\sigma_l) - P(B)|\suchthat B \in \sigma'_{\ell+k}, \ell \ge 1\}
\end{align*}
and if 
$\beta_k = O(k^{-r_{\b}})$ as $k\to\infty$,
$r_{\b}$ is called the \emph{$\b$-mixing exponent}. Note that a larger $r_\b$ guarantees faster mixing. 
We can then restate the main result of~\citet{Yu1994} for the case when $\Hc$ has a finite VC dimension~\cite[Chapter 5]{Shalev-Shwartz--Ben-David2014}.

\begin{theorem}[{\citealp[Corollary 3.2 and Remark (i)]{Yu1994}}]
Assume that a class of binary functions $\Hc$ is of finite VC dimension.
Let $Z^n$ be a stationary $\b$-mixing sequence with $\b$-mixing exponent $r_\b\in (0,1]$. 
Then, for any given $s\in(0,r_{\b})$, we have
\[
n^{s/(1+s)}\frac{\rho_{\Hc}(Z^n)}{n} \stackrel{p}{\longrightarrow} 0 \quad\text{as $n\to\infty$},
\numberthis
\label{eq:BetaMixingVCRate}
\]
where $\stackrel{p}{\longrightarrow}$ denotes convergence in probability.
\end{theorem}
Note that~\eqref{eq:BetaMixingVCRate} immediately implies that $\frac{1}{n} \RegPor \stackrel{p}{\longrightarrow} 0$, \ie $\phi(q_{\Gc}^*)$ is universal in probability. Converting \eqref{eq:BetaMixingVCRate} into a guarantee for convergence in expectation as required in \eqref{eq:ConsistencyCondition} needs an additional argument. In particular, we can show that $\E[\rho_{\Hc}(Z^n)] = O(n^{(3+r_{\b})/(3+2r_{\b}}))$; see Appendix~\ref{app:ex:beta_mixing} for a detailed proof.
\end{example}

\begin{example}[Market history $z_t = \xb^{t-1}_{t-k}$]
\label{ex:canonical}
A canonical example of side information is the market history $z_t=\xv^{t-1}$ or a truncated version of it with memory size $k$, \ie $z_t = \xv^{t-1}_{t-k}$. 
In this case, if the stock market $(\xv_t)$ itself is $k$-th order Markov, then under an additional mild regularity condition, we can show a faster rate $\ERegPor \le \widetilde{O}(\sqrt{n})$ than that implied by the previous example; see Appendix~\ref{app:example_three} for the statement.
\end{example}

\section{Proofs}
\label{sec:proofs}
In this section, we prove Theorem~\ref{thm:AsymptoticConsistency}.
We first note that the probability assignment $q_{\Gc}^*$ used to derive the proposed portfolio guarantees the following regret bound.
\begin{theorem}\label{thm:prob_assign}
For the probability assignment $q_{\Gc}^*$, if the Natarajan dimension $\mathrm{Ndim}(\Gc)=d$ of $\Gc$ is finite and $Z^n \sim P_{Z^n}$ is stationary, we have\footnote{Here, $\log n$ is assumed to be an integer for simplicity, which can be easily rectified at the cost of an absolute constant factor in the regret; see Section~\ref{subsec:ProbAssgnDefn}.}
\begin{align}
&\E\Bigl[\sup_{g\in\Gc}  \sup_{p\in\Pc_S^{\otimes}} \sup_{y^n \in [m]^n} \log \frac{p(y^n\|g(Z^n))}{q_{\Gc}^*(y^n\|Z^n)} \Bigr] 
\le S(d+m)(\log^2n) + 
2.5Sm \sum_{j=0}^{ \log n-1 } j\E[\rho_{\Gc \times \Gc}(Z^{2^{j}})].\nonumber
\end{align}
\end{theorem}
We will first prove Theorem~\ref{thm:prob_assign} and Theorem~\ref{thm:AsymptoticConsistency} then follows almost as a corollary of Theorem~\ref{thm:prob_assign} via the established connection between a probability and the induced portfolio in Proposition~\ref{thm:PortfolioIsEasierSideInfo}.

\subsection{Proof of Theorem~\ref{thm:prob_assign}}
Note that the key building block of the proposed probability assignment scheme $q_{\Gc}^*$ is $q_{\zt^n}(y^i\|z^i)$ defined in \eqref{eq:empirical_covering_mixture}, the uniform mixture based on a minimal empirical covering of $\Gc$ with respect to $\zt^n$. 
The proof consists of three steps.
In Step 1, we first consider the simplest case where the whole side information sequence $z^n$ is provided \emph{noncausally} by an oracle, where we can use $z^n$ as $\zt^n$ to build the empirical covering. We then analyze the performance of $q_{\zt^n}(y^i\|z^i)$ for an arbitrary auxiliary sequence $\zt^n$ in Step 2. Finally, in Step 3, we analyze $q_{\Gc}^*$ based on the analysis of $q_{\zt^n}(y^i\|z^i)$.
\subsubsection*{Step 1. Side Information Given Noncausally}
\label{subsec:NonCausalSideInfo}
Suppose that $z^n$ is available noncausally so that it can be used to construct a minimal empirical covering in $q_{z^n}(y^i\|z^i)$ for $i\in[n]$.
First, note that since $|\{(g(z^n)\suchthat g \in \Gc\}| \le S^n$, we can construct an empirical covering $\{g_1,\ldots,g_\ell\}$ of $\Gc$ with respect to $z^n$ with $\ell \le S^n$. 
Assuming Ndim$(\Gc) = d < \infty$, however, we can even do so with $\ell \le (S^2n)^d$ by Natarajan's Lemma~\cite[Lemma 29.4]{Shalev-Shwartz--Ben-David2014}. 
Hence, for the mixture probability assignment $q_{\zt^n}(y^i\|z^i)$ defined in \eqref{eq:empirical_covering_mixture} with $\zt^n\gets z^n$, \ie
\[
q_{z^n}(y^i\|z^i) = \frac{1}{\ell}\sum_{j=1}^{\ell} q_{\Lap;S}(y^i\|g_j(z^i)),
\]
it readily follows that for any $g\in\Gc$,
\begin{align*}
\sup_{p \in P_S^{\otimes}}\sup_{y^n\in[m]^n} \hspace{-0.4em} \log &\frac{p(y^n\|g(z^n))}{q_{z^n}(y^n\|z^n)} 
\le d \log(S^2n) + Sm \log n\numberthis\label{eq:bound_simplest_case}
\end{align*}
by invoking that $\ell \le (S^2n)^d$ and applying the regret bound for the $m$-ary Laplace probability assignment in Lemma~\ref{lem:LaplaceProbAssgnGuarantee} for each state.


\subsubsection*{Step 2. Auxiliary Side Information Given Noncausally}
\label{subsec:AuxSideInfo}

We now analyze the mixture probability $q_{\zt^n}(y^n\|z^n)$ for an arbitrary auxiliary sequence $\zt^n$, possibly being different from $z^n$.
Intuitively, the sequence $\zt^n$ will also reduce the class $\Gc$ to at most $(S^2n)^d$ functions, and 
if $z^n$ and $\zt^n$ are ``not too far apart'', the two reductions each obtained by $z^n$ and $\zt^n$ may be also close. 
The following lemma provides the performance of the mixture probability $q_{\zt^n}(y^n\|z^n)$ with respect to the auxiliary sequence $\zt^n$, capturing the expected gap from the intuition by the Hamming distance (denoted by $\dham$) between $g(z^n)$ and $\gt(z^n)$.

\begin{lemma}\label{lem:AuxSideInfoNoncausally}
For any $\zt^n$, $z^n$, and $g\in\Gc$ with $\mathrm{Ndim}(\Gc)=d<\infty$, we have 
\begin{align}
\sup_{p \in \Pc_S^{\otimes}} \sup_{y^n \in [m]}  \log \frac{p(y^n\|g(z^n))}{q_{\zt^n}(y^n\|z^n)}
&\le d\log(S^2n) + Sm(\log n)(1+2.5d_H(g(z^n),\gt(z^n)))
\nonumber\\&
\le S(\log n)(
d+m+2.5m \dham(g(z^n),\gt(z^n))).
\end{align}
\end{lemma}
Note that setting $\dham(g(z^n),\gt(z^n))=0$ recovers \eqref{eq:bound_simplest_case} as expected.

\begin{proof}
Let $p_{\thb_{1:S}}$ be a state-wise \iid probability assignment characterized by $\thb_{1:S}=(\thb_1,\ldots,\thb_S)\in\Bc^S$, where $\thb_i = (\th_{i1},\th_{i2},\dotsc,\th_{im}) \in \Bc$ for each $i \in [S]$.
For any state function $g\in\Gc$, by definition of the empirical covering, there exists a function $\gt \in \{\gt_1,\dotsc,\gt_{\ell}\}$ such that $\gt(\zt^n) = g(\zt^n)$. Hence, we first have
\begin{align}
&\log \frac{p_{\thb_{1:S}}(y^n\|g(z^n))}{q_{\zt^n}(y^n\|z^n)} 
\le d \log (S^2n) + \log \frac{p_{\thb_{1:S}}(y^n\|g(z^n))}{q_{\Lap;S}(y^n\|\gt(z^n))}.
\label{eq:tmp}
\end{align}
It only remains to analyze $q_{\Lap;S}(y^n\|\gt(z^n))$.
For each $i \in [S]$ and $j \in [m]$, we define $n_i \defeq |t\suchthat g(Z_t) = i|$ and $k_{ij} \defeq |t\suchthat g(Z_t) = i, y_t = j|$. Moreover let $\nt_i, \kt_{ij}$ be defined in a similar way as  $\nt_i \defeq |t\suchthat \gt(Z_t) = i|$ and $\kt_{ij} \defeq |t\suchthat \gt(Z_t) = i, y_t = j|$). 
We can then write
\[
p_{\thb_{1:S}}(y^n\|g(z^n)) = \prod_{s=1}^S \th_{s1}^{k_{s1}}\dotsc \th_{sm}^{k_{sm}}.
\]

Further, we note that we can explicitly write the expression for the Laplace probability assignment as
\begin{align*}
q_{\Lap}(y^n) = \left(\binom{n+m-1}{m-1} \binom{n}{k_1,\dotsc,k_m}\right)^{-1},
\end{align*}
where $k_i = |\{t\suchthat y_t = i\}|$, and thus its state-wise extension as
\begin{align*}
q_{\Lap;S}(y^n\|\gt(z^n)) 
&= \Bigl(\prod_{s=1}^S \binom{\nt_s+m-1}{m-1}\binom{\nt_s}{\kt_{s1},\dotsc,\kt_{s,m-1}}\Bigr)^{-1}.
\end{align*}
Now, consider
\begin{align}
\log \frac{p_{\thb_{1:S}}(y^n\|g(z^n))}{q_{\Lap;S}(y^n\|\gt(\zt^n))}
&= \sum_{i=1}^S \log \binom{\nt_i+m-1}{m-1}\binom{\nt_i}{\kt_{i1},\dotsc,\kt_{i,m-1}}\th_{i1}^{k_{i1}}\dotsc \th_{im}^{k_{im}} \nonumber\\ 
&\le Sm \log n + \sum_{i=1}^S \log \frac{\binom{\nt_i}{\kt_{i1},\dotsc,\kt_{i,m-1}}}{\binom{n_i}{k_{i1},\dotsc,k_{i,m-1}}} \label{eq:multinomleone}\\
&= Sm \log n + \sum_{i=1}^S \log \frac{\nt_i!}{n_i!} + \sum_{i=1}^S \sum_{j = 1}^m \log \frac{k_{ij}!}{\kt_{ij}}, \label{eq:UseFactorialIneq}
\end{align}
where~\eqref{eq:multinomleone} follows since 
\[
\binom{n_i}{k_{i1},\dotsc,k_{i,m-1}} \th_{i1}^{k_{i1}}\dotsc \th_{im}^{k_{im}} \le 1.
\]

Now, since for all $i \in [S]$ and $j \in [m]$, we have $|n_i - \nt_i| \le \dham(g(z^n),\gt(z^n))$ and $|k_{ij}-\kt_{ij}| \le \dham(g(z^n),\gt(z^n))$, we have that $\nt_i \le n_i + \dham(g(z^n),\gt(z^n))$ and consequently $\frac{\nt_i!}{n_i!} \le \frac{(n_i+\dham(g(z^n),\gt(z^n)))!}{n_i!}$. Thus, we can invoke the exact same calculations as in~\cite[Propositions 5 and 6]{Bhatt--Kim2021} to bound the second and third terms in~\eqref{eq:UseFactorialIneq} as
\begin{align}
\log \frac{p_{\thb_{1:S}}(y^n\|g(z^n))}{q_{\Lap;S}(y^n\|\gt(\zt^n))} 
&\le Sm \log n + S(m+3)\dham(g(z^n),\gt(z^n))\log n \nonumber \\
&\le Sm(\log n)(1+2.5d_H(g(z^n),\gt(z^n))),
\label{eq:HammDistRedundancy}
\end{align} 
since $m\ge 2$.
Plugging this into \eqref{eq:tmp} establishes the first bound. The second bound follows by observing $\log(S^2 n) \le S\log n$.
\end{proof}

When $Z^n$ is stationary as a stochastic process and if $\Zt^n$ is a statistical copy of $Z^n$, the following lemma shows that the Hamming distance can be bounded by $\rho_{\Gc\times\Gc}(Z^n)$, which can be controlled in expectation as $o(n/\log^2 n)$ under mild regularity conditions on $P_{Z^n}$ and $\Gc$.
\begin{lemma}\label{lem:stationary_hamming}
If $Z^n$ is stationary, $\Zt^n\stackrel{(d)}{=}Z^n$, and $\gt(\Zt^n)=\gt(\Zt^n)$, then
\[\dham(g(Z^n),\gt(Z^n))\le \rho_{\Gc \times \Gc}(Z^n) + \rho_{\Gc \times \Gc}(\Zt^n).\]
\end{lemma}
\begin{proof}
Note that for any $Z^n$ and $\Zt^n$, we can write
\begin{align}
\dham(g(Z^n),\gt(Z^n))
&= \dham(g(Z^n),\gt(Z^n)) - \dham(g(\Zt^n),\gt(\Zt^n)) \label{eq:ByDesignZt} \\
&\le \sup_{g_1,g_2} \left| \dham(g_1(Z^n),g_2(Z^n)) - \dham(g_1(\Zt^n),g_2(\Zt^n)) \right| \nonumber \\
&\le \sup_{g_1,g_2} \left| \dham(g_1(Z^n),g_2(Z^n)) - n\P(g_1(Z_1) \neq g_2(Z_2))\right|  \nonumber\\
&\quad+ \sup_{g_1,g_2} \left| \dham(g_1(\Zt^n),g_2(\Zt^n)) - n\P(g_1(\Zt_1) \neq g_2(\Zt_1))\right| \nonumber\\
&=  \rho_{\Gc \times \Gc}(z^n) + \rho_{\Gc \times \Gc}(\zt^n) \label{eq:UseStationarityOfZ}
\end{align}
where~\eqref{eq:ByDesignZt} follows since $\dham(g(\Zt^n),\gt(\Zt^n)) = 0$ by design and \eqref{eq:UseStationarityOfZ} follows since by stationarity of $Z^n \stackrel{(d)}{=} \Zt^n$, we have $n\P(g_1(Z_1) \neq g_2(Z_1)) = n\P(g_1(\Zt_1) \neq g_2(\Zt_1)) = \sum_{i=1}^n \P(g_1(\Zt_i) \neq g_2(\Zt_i)) =  \sum_{i=1}^n \E[\indi{g_1(\Zt_i) \neq g_2(\Zt_i)}]$. Finally, substituting~\eqref{eq:UseStationarityOfZ} into~\eqref{eq:HammDistRedundancy} yields the lemma.
\end{proof}

\subsubsection*{Step 3. Side Information Given Causally}
\label{subsec:ProbAssgnDefn}
In view of Lemma~\ref{lem:stationary_hamming}, provided that $Z^n$ is stationary, we can \emph{bootstrap} the history sequence to construct such an auxiliary sequence, which motivates the epoch-based construction of $q_{\Gc}^*$.
That is, we split the $n$ time steps into $ \log n $ epochs\footnote{For simplicity, we assume that $\log n$ is an integer; if not, we may ``extend'' the horizon of the game from $n$ to $2^{\lceil \log n \rceil} < 2n$, and follow the same analysis incurring at most a constant factor extra in the regret bound.}, and define the $j$-the epoch to consist of the time steps $2^{j-1} + 1 \le i \le 2^j$ starting from $j = 1$, while we define $q_{\Gc}^*(\cdot|Z_1) = 1/m$ for the 0-th epoch. 
For $i \ge 2$, if the time step $i$ falls within the $j$-th epoch, \ie $2^{j-1} + 1 \le i \le 2^j$, then
\begin{align} \label{eq:QStarInJEpochDef}
q_{\Gc}^*(y_i|y^{i-1};Z^i) = \frac{q_{Z^{2^{j-1}}}(y_{2^{j-1}+1}^{i}\|Z_{2^{j-1}+1}^i)}{q_{Z^{2^{j-1}}}(y_{2^{j-1}+1}^{i-1}\|Z_{2^{j-1}+1}^{i-1})}
\end{align}
where we can recall the definition of $q_{Z^{2^{j-1}}}$ from \eqref{eq:empirical_covering_mixture}. 
For any $p\in \Pc_S^{\otimes}$, we then have
\begin{align}
\sum_{i=1}^n \log \frac{p(y_i|g(Z_i))}{q_{\Gc}^*(y_i|y^{i-1};Z^i)}
&\le  \sum_{i=2}^n \log \frac{p(y_i|g(Z_i))}{q_{\Gc}^*(y_i|y^{i-1};Z^i)} + \log m \nonumber \\
&= \sum_{j=1}^{ \log n}\sum_{i=2^{j-1}+1}^{2^j}  \log \frac{p(y_i|g(Z_i))}{q_{\Gc}^*(y_i|y^{i-1};Z^i)}  \nonumber\\
&= \sum_{j=1}^{\log n} \log \frac{p(y_{2^{j-1}+1}^{2^j}\|g(Z_{2^{j-1}+1}^{2^j}))}{q_{Z^{2^{j-1}}}(y_{2^{j-1}+1}^{2^j}\|Z_{2^{j-1}+1}^{2^j})} \label{eq:QStarIsQaux}\\
&\le S(d+m) (\log^2n) 
+ 2.5Sm \sum_{j=0}^{\log n-1} jd_H(g(Z_{1}^{2^{j}}), g(Z_{2^j+1}^{2^{j+1}})),
\label{eq:UseAuxSideInfoLem}
\end{align}
where~\eqref{eq:QStarIsQaux} follows by~\eqref{eq:QStarInJEpochDef} and~\eqref{eq:UseAuxSideInfoLem} follows from Lemma~\ref{lem:AuxSideInfoNoncausally}. 
Finally, taking supremum over $y^n, p$ and $g$ and expectation over $Z^n$ leads to the desired inequality by Lemma~\ref{lem:stationary_hamming}.\qed

\subsection{Proof of Theorem~\ref{thm:AsymptoticConsistency}}
By Proposition~\ref{thm:PortfolioIsEasierSideInfo} and Theorem~\ref{thm:prob_assign}, we have
\begin{align*}
\ERegPor(\phi(q_{\Gc}^*), \Ac_S^{\mathsf{CRP}}, \Gc)
&=\E[\RegPor_n(\phi(q_{\Gc}^*), \Ac_S^{\mathsf{CRP}}, \Gc;\Xv^n,Z^n)]\\
&\le \E[\RegProb_n(q;\Pc,\Gc;Z^n)]\\
&= \E\Bigl[\sup_{g\in\Gc} \sup_{p\in\Pc} \max_{y^n} \log \frac{p(y^n\|g(Z^n))}{q(y^n\|Z^n)}\Bigr]\\
&\le S(d+m)(\log^2n) +2.5Sm\sum_{j=0}^{\log n-1} j\E[\rho(Z^{2^j})],
\end{align*}
where we omit the subscript in $\rho_{\Gc \times \Gc }(\cdot)$ for brevity. 
Since the first term in the bound is sublinear in $n$ when $d$ and $S$ are fixed, it then suffices to show that 
\begin{align}\label{eq:AsymptoticRegTwoTerms}
\sum_{j=0}^{ \log n -1} j\E[\rho(Z^{2^{j}})] =o(n).
\end{align} 
Using the change of variables $n' = \log n$, observe
\begin{align}
\sum_{j=0}^{ \log n -1} j\E[\rho(Z^{2^{j}})] 
&= \frac{1}{n'} \sum_{j=0}^{n'-1}
j\E[\rho(Z^{2^{j}})]\frac{n'}{2^{n'}} 
\le \frac{1}{n'} \sum_{j=0}^{n'-1}  \frac{j^2}{2^j}\E[\rho(Z^{2^j})],
\nonumber
\end{align}
where the inequality follows since $\frac{n'}{2^{n'}} \le \frac{j}{2^j}$ for all $j \le n'$. Now, since $\frac{(\log n)^2}{n}\E[\rho(Z^n)]=\frac{n'^2}{2^{n'}} \E[\rho(Z^{2^{n'}})] \to 0$ as $n\to\infty$ is assumed, we also have $\frac{1}{n'} \sum_{j=0}^{n'-1} \frac{j^2}{2^j}\E[\rho(Z^{2^j})] \to 0$ as $n' \to \infty$, by the Ces\`aro mean Theorem. A final change of variables concludes the proof. \qed

\section{Related Work and Discussion}
\label{sec:related}
Portfolio selection has been a closely studied topic in information theory since the seminal work of~\cite{Cover1991} and~\cite{Cover--Ordentlich1996},  both of which established close connections  between portfolio selection and the classically studied information theoretic problem of universal compression~\citep{Rissanen1996, Lempel--Ziv1978, Feder--Merhav1998, Xie--Barron2000}. A number of variations have been considered since, for example incorporating transaction costs~\citep{Blum--Kalai99, Uziel--El-Yaniv2020} using other probability assignments than \iid~\citep{Kozat--Singer--Bean2008, Tavory--Feder2010}, and considering space complexity issues~\citep{Tavory--Feder2008}.

\cite{Cross--Barron2003} and~\cite{Gyorfi--Lugosi--Udina2006} proposed  portfolio selection techniques incorporating continuous side information; however, the competitor classes considered in both are disparate from ours making the problems different. 

As demonstrated, portfolio selection with side information is closely related to sequential prediction with side information and log-loss. This problem has attracted recent interest~\citep{Rakhlin--Sridharan2015, Bilodeau--Foster--Roy2020, Fogel--Feder2017, Bhatt--Kim2021}, with the first two focused on obtaining fundamental limits via the sequential complexities approach of~\cite{Rakhlin--Sridharan--Tewari2015}. 
More recently, the preprint of~\citet{Bilodeau--Foster--Roy2021} 
proposed a mixture-based conditional density estimator,
which specifically achieves $\E[\RegProb]=O(\log^2 n)$ for the binary probability assignment problem with \iid side information with a VC class, which tightens the regret $\tilde{O}(\sqrt{n})$ established in \citep{Bhatt--Kim2021}.
Therefore, it is natural to consider applying the probability assignment of \citet{Bilodeau--Foster--Roy2021} in hoping to relax the technical condition~\eqref{eq:ConsistencyCondition}.
We note, however, that analyzing their method in our setting of non-\iid side information sequences seems to involve a significant amount of additional work. More precisely, their analysis needs to be extended to (1) individual-sequence $y^n$ and (2) stationary ergodic side information with a dependence of the regret on $\rho_{\Hc}(Z^n)$ similar to that of the method of~\citet{Bhatt--Kim2021}. In their words, we would need to relax the assumption of the data being \emph{well-specified}.

At a high level, they use a similar covering approach (with respect to the Hellinger metric over distributions) as well as a \emph{smoothing} of probabilities in order to avoid unbounded likelihood ratios (we, in contrast, have used the Laplace/KT probability assignment). Using a similar epoch-based analysis they establish regret bounds in \citep[Appendix~D]{Bilodeau--Foster--Roy2021} by first upper bounding the KL divergence in terms of the Hellinger divergence and then leveraging local Rademacher complexities in conjunction with an inequality of~\citet{Bousquet2002}. 
In order to extend their method to individual-sequence $y^n$ and stationary ergodic $Z^n$, one would need to either extend the aforementioned inequality to these cases, or to bypass the step of upper-bounding the KL divergence in terms of the Hellinger divergence altogether. 
We leave these directions of extension for future work.



\appendix

\section{Deferred Proofs}

\subsection{Proof of Propostions~\ref{prop:portfolio_probability_side_information} and~\ref{thm:PortfolioIsEasierSideInfo}}
It suffices to prove Proposition~\ref{thm:PortfolioIsEasierSideInfo}, since Proposition~\ref{prop:portfolio_probability_side_information} follows from it by taking $z^n = w^n$ and taking $|\Gc| = 1$ with the function $g \in \Gc$ being simply $g(z) = z$.  

Recall that for a probability assignment $q(y_i|y^{i-1};z^i)$, we have the probability induced portfolio $a = \phi(q)$ defined as 
\[
a(j|\xv^{t-1};z^t) \defeq \frac{\sum_{y^{t-1}} q(y^{t-1}j\|z^{t})\xv(y^{t-1})}{\sum_{y^{t-1}} q(y^{t-1}\|z^{t-1})\xv(y^{t-1})},
\]
where recall for $t \in [n], q(y^t\|z^t) = \prod_{i=1}^t q(y_i|y^{i-1};z^i)$. We then have 
\[
\sum_{y_t\in[m]} a(y_t|\xv^{t-1};z^t) \xv_t(y_t) = \frac{\sum_{y^{t-1}} q(y^{t}\|z^{t})\xv(y^{t})}{\sum_{y^{t-1}} q(y^{t-1}\|z^{t-1})\xv(y^{t-1})},
\]
and consequently using a telescoping argument, 
\[
S_n(\phi(q),\xv^n;z^n) = \sum_{y^n \in [m]^n} q(y^n\|z^n) \xv(y^n).
\]
From this we can see that
\begin{align}
\RegPor_n(\phi(q);\phi(\Pc),\Gc;\xv^n,z^n) &= \sup_{g\in\Gc}\sup_{p\in\Pc}
 \log\frac{S_n(\phi(p),\xv^n;g(z^n))}{S_n(\phi(q),\xv^n;z^n)} \nonumber\\
 &= \sup_{g\in\Gc}\sup_{p\in\Pc} \log\frac{\sum_{y^n \in [m]^n} p(y^n\|g(z^n))\xv(y^n)}{\sum_{y^n \in [m]^n} q(y^n\|z^n)\xv(y^n)} \nonumber\\
 &\le \sup_{g \in \Gc} \sup_{p \in \Pc} \max_{y^n \in [m]^n } \frac{p(y^n\|g(z^n))}{q(y^n\|z^n)} \label{eq:BasicLemma} \\
 &= \RegProb_n(q;\Pc,\Gc;z^n),\nonumber
\end{align}
where~\eqref{eq:BasicLemma} follows from Lemma~\ref{lem:ratio}.\qed


\subsection{Proof for \texorpdfstring{$\E[\rho_{\Hc}(Z^n)] = o(\frac{n}{\log n})$}{} in Example~\ref{ex:beta_mixing}}
\label{app:ex:beta_mixing}

Let $Z^n$ be a $\b$-mixing process with $\b$-mixing coefficient $\b_k$ and $\b$-mixing exponent $r>0$, \ie $\b_k = O(k^{-r})$ as $k\to\infty$. 
Indeed, under this assumption, we can show a stronger result $\E[\rho_{\Hc}(Z^n)] = O(n^{\frac{3+r}{3+2r}})$ following a similar argument as in \citet{Karandikar--Vidyasagar2002} and \citet{Hanneke--Yang2019}.

Pick $k\ge 1$ which divides $n$ for simplicity; the divisibility can be easily lifted by elongating the game from $n$ steps to the next number divisible by $k$. 
We will choose $k$ as a function of $n$ at the end of proof.
We define the nonoverlapping $k$ subsequences $Z^{(1)},\dotsc,Z^{(k)}$ of length $n/k$ as 
\begin{align*}
    {Z^{(1)}}_{1}^{n/k} &= Z_1,Z_{k+1},Z_{2k+1}\dotsc,Z_{(n/k -1)+1}, \\
    {Z^{(2)}}_{1}^{n/k} &= Z_2,Z_{k+2},Z_{2k+2}\dotsc,Z_{(n/k -1)k+2}, \\
    &\vdots\\
    {Z^{(k)}}_{1}^{n/k} &= Z_{k},Z_{2k},Z_{3k}\dotsc,Z_{(n/k)k}.
\end{align*}

We will invoke the classical result on $\b$-mixing processes that states that
\begin{align}\label{eq:TVDistBetaMixing}
    \dtv\biggl(P_{{Z^{(j)}}_{1}^{n/k}}, \prod_{i=1}^{n/k} P_{Z^{(j)}_i}\biggr) \le \left(\frac{n}{k}-1\right)\b_k
\end{align}
for each $j \in [k]$, where $\dtv(\cdot, \cdot)$ denotes the total variation distance; see, for example,~\cite[Lemma 1]{Hanneke--Yang2019} and the references therein.

Now, we consider
\begin{align}
    \E[\rho_{\Hc}(Z^n)]
    &= \E\Bigl[\sup_{h \in \Hc} \Bigl|\sum_{i=1}^n (h(Z_i)-E[h(Z_i)])\Bigr|\Bigr] \nonumber\\
    &\le \E\biggl[\sum_{j=1}^{k} \sup_{h \in \Hc} \Bigl|\sum_{i=1}^{n/k}(h(Z^{(j)}_i)-\E[h(Z^{(j)}_i)])\Bigr|\biggr] \nonumber\\
    &= \sum_{j=1}^{k} \E\biggl[\sup_{h \in \Hc} \Bigl|\sum_{i=1}^{n/k}(h(Z^{(j)}_i) - \E[h(Z^{(j)}_i)])\Bigr|\biggr].\label{eq:kSubsequences}
\end{align}
Let $Z'_1,\dotsc,Z'_{n/k}$ be an \iid process with the same marginal distribution of the stationary process $Z^n$, \ie $P_{Z'_1} = P_{Z_1}$. 
Continuing from the summand in \eqref{eq:kSubsequences}, we then have
\begin{align}
    &\E\biggl[\sup_{h \in \Hc} \Bigl|\sum_{i=1}^{n/k}(h(Z^{(1)}_i) - \E[h(Z^{(1)}_i)])\Bigr|\biggr] \\
    &=  \E\biggl[\sup_{h \in \Hc} \Bigl|\sum_{i=1}^{n/k}(h(Z^{(1)}_i) - h(Z'_i) + h(Z'_i) - \E[h(Z^{(1)}_i)])\Bigr|\biggr] \nonumber\\
    &= \E\biggl[\sup_{h \in \Hc} \Bigl|\sum_{i=1}^{n/k}(h(Z^{(1)}_i) - h(Z'_i) + h(Z'_i) - \E[h(Z'_i)])\Bigr|\biggr] \label{eq:ZPrimeSameDist} \\
    &\le \E\biggl[\sup_{h \in \Hc} \Bigl|\sum_{i=1}^{n/k}(h(Z^{(1)}_i) - h(Z'_i))\Bigr|\biggr] + \E\biggl[\sup_{h \in \Hc} \Bigl|\sum_{i=1}^{n/k}(h(Z'_i) - E[h(Z'_i)])\Bigr|\biggr] \\
    &\le \E\biggl[\sup_{h \in \Hc} \Bigl|\sum_{i=1}^{n/k}(h(Z^{(1)}_i) - h(Z'_i))\Bigr|\biggr] + C\sqrt{\frac{dn}{k}} \label{eq:ZPrimeIsIID}\\
    &\le \frac{n}{k}\sup_{h \in \Hc} \E\Biggl|\frac{k}{n}\sum_{i=1}^{n/k} h(Z^{(1)}_i) - \frac{k}{n}\sum_{i=1}^{n/k} h(Z'_i) \Biggr|       + C\sqrt{\frac{dn}{k}} \nonumber \\ 
    &\le \frac{n}{k} \dtv\biggl(P_{{Z^{(j)}}_{1}^{n/k}}, \prod_{i=1}^{n/k} P_{Z^{(j)}_i}\biggr) + C\sqrt{\frac{dn}{k}} \label{eq:TVdefn} \\
    &\le \frac{n^2\b_k}{k^2} + C\sqrt{\frac{dn}{k}}. \label{eq:UseBetaMixingTV}
\end{align}
Here, \eqref{eq:ZPrimeSameDist} follows since the marginal distribution $Z'_i \stackrel{(d)}{=} Z^{(1)}_i$, \eqref{eq:ZPrimeIsIID} follows since the distribution $Z'^n$ is \iid and from~\citep[Theorem 8.3.23]{Vershynin2018},~\eqref{eq:TVdefn} follows from the following variational form of the total variation distance $\dtv(P,P')$ between two measures $P$ and $P'$ defined over the same measure space, \ie
\[
\dtv(P,P') = \sup_{f\suchthat |f| \le 1}|\E_{X \sim P}[f(X)] - \E_{X \sim P'}[f(X)]|,
\]
and lastly \eqref{eq:UseBetaMixingTV} follows from \eqref{eq:TVDistBetaMixing}. 
Substituting \eqref{eq:UseBetaMixingTV} into \eqref{eq:kSubsequences} yields that 
\[
\E[\rho_{\Hc}(Z^n)] \le \frac{n^2\b_k}{k} + C\sqrt{dnk} \le \frac{C'n^2 k^{-r}}{k} +  C\sqrt{dnk}
\]
for $k$ sufficiently large with some $C'>0$,
where we use the definition of the $\b$-mixing exponent $r$ in the second inequality. 
Finally, choosing $k = O(n^{\frac{3}{3+2r}})$ yields the claimed rate $\E[\rho_{\Hc}(Z^n)] = O(n^{\frac{3+r}{3+2r}})$.
\qed

\section{A Detailed Discussion on Example~\ref{ex:canonical}}
\label{app:example_three}
   

For the side information $z_t=\xv_{t-k}^{t-1}$ in Example~\ref{ex:canonical}, if the market $(\Xv_t)$ itself is $k$-th order Markov, then we can establish the following guarantee.

\begin{lemma}\label{lem:MarkovMarket}
Let $\Xv^n$ be a stationary $k$-th order Markov process and let $Z_t = \Xv^{t-1}_{t-k}\in (\Real_+^m)^k$. 
Suppose that (1) the density of $Z_0=\Xv_{-k}^{-1}$ exists and is bounded and supported over a bounded, convex set $E \subset (\Real^{m}_{+})^k$ with nonempty interior and (2) there exist $b > 0$ and $\eps>0$ such that the time-invariant conditional density satisfies
\[
p_{\Xv_{t-k+1}^{t}|\Xv_{t-k}^{t-1}}(z'|z)\ge b1_{B(z,\eps)}(z')
\]
for any $z\in(\Real_+^m)^k$, where $B(z,\e)$ denotes the open ball of radius $\eps$ centered at $z\in(\Real_+^m)^k$ with respect to Euclidean distance. 
Then, we have $\E[\rho_{\Hc}(Z^n)] = \widetilde{O}(\sqrt{n})$.
\end{lemma}
\begin{proof}
This is a direct consequence of \cite[Proposition 11]{Bertail--Portier2019}, which establishes an upper bound on $\E[\rho_{\Hc}({Z'}^n)]$ for a \emph{Metropolis--Hastings} (MH) walk ${Z'}^n$.
First, note that $Z^n$ forms a Markov chain due to the $k$-th order Markovity of $\Xv^n$.
To apply the proposition over the Markov chain $Z^n$, we set the proposal distribution $q$ in the MH algorithm to be the actual transition kernel of the Markov chain ${Z}^n$, so that the MH walk becomes the process $Z^n$ of our interest.
Then, under the assumptions above, we can apply the result of \citeauthor{Bertail--Portier2019} and conclude that $\E[\rho_{\Hc}({Z}^n)]=\tilde{O}(\sqrt{n})$ for a VC-class $\Hc$.
\end{proof}

\section{Simulation Based on Monte Carlo Approximation}
\label{sec:simulation}
The (maybe the only) downside of the universal portfolio algorithms is their computational complexity. 
It is not hard to see that the \emph{exact} computation of Cover's universal portfolio requires, on the $T$-th day of investment over $m$ stocks, $O(T^m)$ time complexity, and the computation quickly become infeasible for a long investment period; see \citep{Cover--Ordentlich1996} for a detailed argument. 
An efficient implementation of universal portfolios is a decades-old open problem and still remains as an active area of research~\citep{Luo--Wei--Zheng2018,VanErven--VanDerHoeven--Kotlowski--Koolen2020}.
Hence, in this paper, we consider a Monte Carlo simulation of the universal portfolio algorithms based on the cumulative wealth expression of a probability induced portfolio~\eqref{eq:wealth_achieved_by_prob_induced_portfolio}. While it is a very crude approximation for large $m$, $S$, or $\mathrm{Ndim}(\Gc)$, this at least provides a way to demonstrate the performance of the ideas.

First, note that from \eqref{eq:wealth_achieved_by_prob_induced_portfolio}, the cumulative wealth achieved by Cover's universal portfolio $\phi(q_{\Lap})$ can be written as
\[
S_n(\phi(q_{\Lap}), \xb^n)
=\sum_{y^n\in [m]^n} q_{\Lap}(y^n)\xb(y^n)
=\int_{\Bc} S_n(\thb, \xb^n) \mu(\thb) \diff\thb,
\]
since the Laplace probability assignment $q_{\Lap}(y^n)=\int_{\Bc} \mu(\thb)p_{\thb}(y^n)\diff\thb$ is a mixture with respect to a uniform density $\mu(\thb)$ over the simplex $\Bc$.
Hence, if we draw $N$ CRPs $\thb_1,\ldots,\thb_N$ from $\mu$ and \emph{buy-and-hold} uniformly over the CRPs, we will attain approximately similar wealth and the approximation will get better as $N$ becomes larger.
Note, however, that this naive approximation requires $N=\Omega(\frac{1}{\eps^m})$ to achieve an approximation error $\eps$ and thus may not be feasible when the number of stocks $m$ is large.

A similar crude approximation can be performed for the universal portfolio $\phi(q_{\Lap;S})$ with discrete side information $w^n$, since 
\begin{align*}
S_n(\phi(q_{\Lap;S}),\xb^n;w^n)
&= \sum_{y^n\in[m]^n} q_{\Lap;S}(y^n\|w^n) \xb(y^n)
= \prod_{s=1}^S S_{|\xb^n(s;w^n)|}(\phi(q_{\Lap}), \xb^n(s;w^n)),
\end{align*}
where $\xb^n(s;w^n)=(\xb_i\suchthat w_i=s,i\in[n])$,
since $q_{\Lap;S}(y^n\|w^n) = \prod_{s=1}^S q_{\Lap}(y^n(s;w^n))$.
That is, we can crudely approximate the performance of $\phi(q_{\Lap;S})$ by simply drawing many state-wise CRPs $\thb_{1:S}$ according to $\mu(\thb_{1:S})\defeq \mu(\thb_1)\cdots\mu(\thb_S)$ and running the buy-and-hold strategy.

We can now consider an approximation of the proposed strategy $\phi(q_{\Gc}^*)$.
By the epoch-wise construction of $q_{\Gc}^*$ as explicitly shown in \eqref{eq:assigned_probability}, the cumulative wealth can be also factorized as
\begin{align*}
S_n(\phi(q_{\Gc}^*),\xb^n;z^n)
&= \prod_{j=1}^{J} \sum_{y_{2^{j-1}+1}^{2^j} \in[m]^{2^{j-1}}} q_{\Gc;z^{2^{j-1}}}(y_{2^{j-1}+1}^{2^j} \| z_{2^{j-1}+1}^{2^j})
\xb_{2^{j-1}+1}^{2^j}(y_{2^{j-1}+1}^{2^j})\\
&= \prod_{j=1}^J S_{2^{j-1}}(\phi(q_{\Gc;z^{2^{j-1}}}), \xb_{2^{j-1}+1}^{2^j}; z_{2^{j-1}+1}^{2^j}).
\end{align*}
where we assume $n=2^J$ for simplicity.
Here, for each $j\in[J]$, if $\{\gt_1,\ldots,\gt_{\ell_j}\}$ is a minimal empirical covering of $\Gc$ with respect to $z^{2^{j-1}}$, we can write
\begin{align*}
S_{2^{j-1}}(q_{\Gc;z^{2^{j-1}}}, \xb_{2^{j-1}+1}^{2^j}; z_{2^{j-1}+1}^{2^j})
&= \frac{1}{\ell_j} 
\sum_{k=1}^{\ell_j}
S_{2^{j-1}}(\phi(q_{\Lap;S}), \xb_{2^{j-1}+1}^{2^j}; \gt_k(z_{2^{j-1}+1}^{2^j})).
\end{align*}
For each state function $\gt_k$, the summand is the cumulative wealth of the UP with the side information $\gt_k(z_{2^{j-1}+1}^{2^j})$ and thus can be approximated by the same argument from the previous paragraph. 

This leads to the following Monte Carlo simulation of the proposed algorithm. Let $N$ be the number of Monte Carlo samples used in the approximation.
\vspace{.5em}

\noindent\fbox{\begin{minipage}{\textwidth}
For each epoch $j=1,2,\ldots$:
\begin{enumerate}
\item Find an empirical covering $\{\gt_1,\ldots,\gt_{\ell_j}\} \subseteq \Gc$ with respect to $z^{2^{j-1}}$.
\item For each $k\in [\ell_j]$, draw $N$ state-wise CRPs $(\thb_{1:S,i}^{(k)})_{i=1}^N$ from $\mu(\thb_{1:S})$ at random.
\item During the $j$-th investment epoch, \ie $t\in(2^{j-1}, 2^j]$, run the buy-and-hold strategy uniformly over all sampled CRPs $(\thb_{1:S,i}^{(k)})_{i=1}^N$ for each $k\in [\ell_j]$.
\item At the end of the epoch, sell all stocks.
\end{enumerate}
\end{minipage}}\vspace{.5em}

We stress that this only simulates the cumulative wealth of the universal portfolio algorithm by directly estimating the cumulative wealth expression, rather than approximating actions of the algorithm for each round.
Note that a more sophisticated Monte Carlo Markov Chain based approximation for Cover's universal portfolio was proposed and analyzed by \citet{Kalai--Vempala2002}.
It is left as a future direction to extend their method for our algorithm with continuous side information.

In the following, we study a simple example for concreteness, which admits an easy construction of minimal empirical coverings. 
Note that, for a richer class of state functions, finding a minimal empirical covering may be another computational bottleneck.
\begin{example}
As a simple case of the canonical side information considered in Example~\ref{ex:canonical}, we choose the price relative of the stock 1 on the previous day as the continuous side information, \ie $z_t=\xv_{t-1,1}$, and a class of 1D threshold functions $\Gc=\{x\mapsto g_a(x)=1\{x\ge a\}\suchthat a>0\}$ of $\mathrm{Ndim}(\Gc)=1$. 
Note that we consider a binary state space ($S=2$). 
In this case, it is easy to show that $\{g_{x_{0,1}},\ldots, g_{x_{t-1,1}}\}$ is a minimal empirical covering given $z^t=(x_{i,1})_{i=0}^{t-1}$.

In general, we can consider $z_t=\xv_{t-1}$ with a class of product of 1D threshold functions $\Gc=\{x\mapsto g_{\mathbf{a}}(\xb)=(1\{x_1\ge a_1\},\ldots,1\{x_m\ge a_m\})\suchthat \mathbf{a}=(a_1,\ldots,a_m) \in\Real_{++}^m\}$ of $\mathrm{Ndim}(\Gc) \le m \log m$~\citep[Lemma 29.6]{Shalev-Shwartz--Ben-David2014} and $S=2^m$.
Given $z^t=\xv^{t-1}$, $\{g_{\xv_0},\ldots,g_{\xv_{t-1}}\}$ is a minimal empirical covering.
\end{example}

We briefly demonstrate how the proposed portfolio performs on two real stocks.
We collected the 6-year period from Jan-01-2012 to Dec-31-2017 (total 1508 trading days) of two stocks Ford (F) and Macy's (M).
Over the period, Ford went up by a factor of 1.11, while Macy's went down by a factor of 0.77.
The best CRP in hindsight, which turns out to be the buy-and-hold of Ford, achieves a growth factor of 1.11. The uniform CRP achieves a growth factor of 0.99.
While the universal portfolio without side information achieves a growth factor of only, the proposed algorithm with the yesterday's prices and the class of thresholding functions achieves a growth factor of 1.15.

We note that there can exist more sophisticated, carefully chosen side information and state-function classes that may exhibit better performance in practice than the simple example above. 
We leave the problem of constructing good continuous side information and extensive experiments as future work.

\bibliographystyle{plainnat}
\bibliography{ref.bib}

\end{document}